\documentclass{amsart}
\usepackage{amsmath}
\usepackage{amssymb}
\usepackage{amsthm}
\usepackage{tikz}
\usepackage[utf8]{inputenc}
\usepackage[english]{babel}
\usepackage{microtype}

\newtheorem{theorem}{Theorem}

\theoremstyle{plain}

\newtheorem{lemma}{Lemma}

\newcommand{\0}{{\tt 0}}
\newcommand{\1}{{\tt 1}}

\newcommand{\abs}[1]{\left|#1\right|}

\newcommand{\pref}{\leq_p}
\newcommand{\suff}{\leq_s}

\newcommand{\analyze}[1]{\noindent {\small \bf #1}}

\newcommand{\less}{\vartriangleleft}
\newcommand{\LESS}{\blacktriangleleft}

\begin{document}

\title{Fully bordered words}
\begin{abstract}
    We characterize binary words that have exactly two unbordered conjugates and show that they can be expressed as a product of two palindromes.
\end{abstract}
\author{\v St\v ep\' an Holub}
\address{Department of Algebra, Charles University, Sokolovsk\'a 83, 175 86 Praha, Czech Republic}
\thanks{Supported by the Czech Science Foundation grant number 13-01832S}
\email{holub@karlin.mff.cuni.cz}
\author{Mike M\"uller}
\address{Institut f{\"u}r Informatik, Christian-Albrechts-Universit{\"a}t zu Kiel, D-24098 Kiel, Germany}
\email{mail@mikemueller.name}
\subjclass{68R15}
\keywords{palindromes, unbordered words, Lyndon words}

\maketitle

It is well known that each primitive word has an unbordered conjugate. The typical example of such a conjugate is the Lyndon conjugate, that is, the minimal one in a lexicographic ordering. Using different lexicographic orders, we can obtain several unbordered conjugates. In particular, any primitive binary word has at least two unbordered conjugates. This leads to a natural question about the structure of words that do not have any other unbordered conjugate apart from these obligatory two. We shall call them \emph{fully bordered}.  In this paper we give an inductive characterization of fully bordered words, which shows that they have some kind of fractal structure. One consequence of this characterization is that each fully bordered word is a product
of two palindromes. This result can be also interpreted in terms of palindromic length, introduced in \cite{pallength}, saying that fully bordered words have the palindromic length two. The property was conjectured in 2012 by Luca Zamboni (personal communication), and it was the motivation for our research.

The other extremal case, that is, binary words with as many unbordered conjugates as possible, was studied in \cite{bordercorr}. The present work is also related to research on critical points of a word, which is a stronger concept than unbordered conjugate: the conjugate in the critical point is always unbordered, but not vice versa. Some results related to our words can be found in \cite{density}, where authors investigate words with many and few critical points. In particular, they study words with a unique critical point. A related topic is also Duval's problem and the Ehrenfeucht-Silberger problem, solved in \cite{DuvalDirk, Duval, EhrSil} using similar techniques as those we employ in this paper.

\section{Preliminaries}
We first review basic concepts and facts we use in this paper. Let $\Sigma = \{\0, \1\}$ be a binary alphabet and $\Sigma^*$ the free monoid generated by $\Sigma$ using the concatenation operation with the empty word as the unit.
We refer to elements of $\Sigma^*$ as \emph{binary words}.
The concatenation of two words $u, v$ is denoted by $u \cdot v$, but we will omit the operator most of the times and simply write $uv$. 
We say that a word $w = w_1w_2\cdots w_n$, where $w_i \in \Sigma$ for all $1 \leq i \leq n$, is of \emph{length} $n$,
and we denote this as $|w| = n$.

The two possible lexicographic orders on $\Sigma^*$ are $\less$, defined by $\0 \less \1$, and $\LESS$,
where $\1 \LESS \0$.
The \emph{reversal} of a word $w = w_1w_2\cdots w_n$ of length $n$ is the word $w^R = w_n \cdots w_2w_1$.
A word $w$ is called a \emph{palindrome} if $w = w^R$.

Two words $w$ and $w'$ are \emph{conjugates} if there exist words $u$ and $v$, such that $w = uv$ and $w' = vu$.
The set of all conjugates of $w$ is denoted by $[w]$.
For a word $w = uvz$, we say that $u$ is a \emph{prefix} of $w$, $v$ is a \emph{factor} of $w$, and $z$ is a \emph{suffix} of $w$.
We denote these prefix- and suffix-relations by $u \pref w$ and $z \suff w$, respectively.
Furthermore, we write $vz = u^{-1}w$ and $uv = wz^{-1}$ in this case. A prefix (suffix) of $w$ is \emph{proper} if $u\neq w$ ($z\neq w)$; then we write $u<_p w$ ($z<_s w$).
A \emph{cyclic occurrence} of $u$ in $w$ is a number $0 \leq m <|w|$ such that $u$ is a prefix of $w_2w_1$ where $w_1w_2=w$ and $m=|w_1|$.

A word is \emph{primitive} if $w = t^k$ implies $k=1$ and thus $w = t$. It is a basic fact that if $uv=vu$ then $u=t^i$ and $v=t^j$ for some $t$ and some $i,j\geq 0$ (see, for example, section 2.3 of \cite{jeff}). In particular, no primitive word is a nontrivial conjugate of itself.

A word $w$ is \emph{bordered}, if there is a nonempty  word $u\neq w$ that is both a suffix and a prefix of $w$. Any such $u$ is called a \emph{border} of $w$. Note that a border may be of length greater than $|w|/2$. However, it is easy to see that any bordered word has a border $u$ such that $w=uvu$ for some (possibly empty) word $v$. Naturally, a word is \emph{unbordered}, if it is not bordered.

If  $w = w_1w_2\cdots w_n$ and $w_{i+p} = w_i$ holds for some $p\geq 1$ and all $1 \leq i \leq n-p$, then $p$ is \emph{a period} of $w$. The smallest period of $w$ is called \emph{the period} of $w$. The prefix $u$ of $w$ such that $|u|$ is the period of $w$ is called the \emph{periodic root} of $w$. Note that the periodic root is always primitive. Note also that a word $w$ is unbordered if and only if its smallest period is $|w|$. In particular, an unbordered word $w$ cannot be a factor of a word with a period smaller than $|w|$.

A \emph{Lyndon word} is a primitive word $w$ that is the lexicographically smallest element (with respect to some order) of $[w]$.
To make it clear which order is considered when we compare words or talk about minimal words, we will speak about $\less$-Lyndon words or $\LESS$-Lyndon words.

 Let $z$, $|z|<|w|$, be a prefix of a word $w$ and let $r$ be the shortest border of the word $w_z:=z^{-1}wz$ if $w_z$ is bordered, and let $r=w_z$ otherwise.
We say that $r$ is the \emph{local periodic root of $w$ at the point $m=|z|$} and $|r|$ is \emph{the local period} at the point $m$. 
We note that our concept of the local period is consistent with the terminology used in the literature if $w$ is considered as a cyclic word. Accordingly, the local periodic root $r$ is the shortest nonempty word for which $rr$ is centered at the given point $|z|$ in a cyclically understood word $w$. 

We stress that the local periodic root is always unbordered (being defined as the shortest border). This implies that the local period of a primitive word $w$ at any point is either less than $|w|/2$ (equality would yield imprimitivity) or $|w|$. In the latter case we say that the local period is \emph{trivial}.

We now prove several basic and useful properties of Lyndon words. For convenience, we formulate them for the order $\less$. 
\begin{lemma}\label{lem:self}
Any Lyndon word is unbordered.
\end{lemma}
\begin{proof}
Let $w$ be a $\less$-Lyndon word. Assume that there is nonempty $u$ such that $w=uvu$. Then $uvu\less uuv$ implies $vu\less uv$, which in turn yields $vuu\less uvu$, a contradiction. 
\end{proof}

\begin{lemma}\label{lem:4}
    Let $w = uv$ be a $\less$-Lyndon word and let $z$ be the periodic root of $u$.
    Then $z$ is a $\less$-Lyndon word.
\end{lemma}
\begin{proof}
Let $u=(z_1z_2)^kz'$ where $z=z_1z_2$ is the periodic root of $u$, $z_1<_p z$ so that $z_2z_1$ is the $\less$-Lyndon conjugate of $z$, and $z'<_p z$. Let $\tilde z$ be the prefix of $u$ of length $|z_1^{-1}u|$. Since $uv$ is $\less$-Lyndon, we have $\tilde z\less z_1^{-1}u$, and since $z_2z_1$ is $\less$-Lyndon, we have $z_1^{-1}u\less \tilde z$. Therefore $\tilde z=z_1^{-1}u$,  which implies that $z_1$ is empty, otherwise $|z_1|$ is a period of $u$.  
\end{proof}

\begin{lemma}\label{lem:5}
    Let $w$ be a $\less$-Lyndon word and let $z^kz' \cdot \1 \pref w$ where $z$ is the periodic root of $z^kz'$, and $z' \cdot \0 \pref z$.
    Then $z^kz' \cdot \1$ is a $\less$-Lyndon word. 
\end{lemma}
\begin{proof}
Suppose that $r\cdot \1$ is a border of $z^kz'\cdot \1$. By Lemma \ref{lem:4} and Lemma \ref{lem:self}, $z$ is unbordered, which implies that $r\cdot \1$ is a suffix of $z'\cdot \1$. Therefore $r\cdot \0$ is a factor of $z$. Since $r\cdot \1$ is a prefix of $w$, we have a contradiction with $w$ being $\less$-Lyndon.

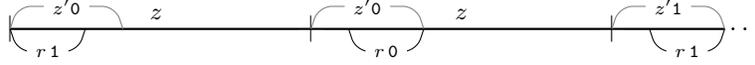
\[
\begin{tikzpicture}
[tecka/.style={circle,inner sep=1pt, fill}
]
\def\d1{0} \def\dz{4} \def\dzz{2*\dz} \def\dzc{1.5} \def\dzp{\dzz+\dzc} \def\dr{1}
\node[inner sep=0] (0) at (\d1,0){${\mathbf \vert}$};
\node (z) at (\dz,0){$\vert$};
\node (zz) at (\dzz,0){$\vert$};
\node[inner sep=0] (zp) at (\dzp,0){};
\draw (0) -- node[above] {{\small $z$}} (z);
\draw (z) -- node[above] {{\small $z$}} (zz);
\draw (zz) -- node[above] {} (zp);
\draw[thick] (0.center) -- (zp.center);
\node at (\dzp+0.3,0) {$\dots$};	
\draw[gray,rounded corners=5pt] (\dz,0) -- (\dz+0.1,0.3) --
node[fill=white] {\textcolor{black}{{\scriptsize $z'\0$}}}  (\dz+\dzc-0.1,0.3) --(\dz+\dzc,0);
\draw[gray,rounded corners=5pt] (\d1,0) -- (\d1+0.1,0.3) --
node[fill=white] {\textcolor{black}{{\scriptsize $z'\0$}}}  (\d1+\dzc-0.1,0.3) -- (\d1+\dzc,0);
\draw[gray,rounded corners=5pt] (\dzz,0) -- (\dzz+0.1,0.3) --
node[fill=white] {\textcolor{black}{{\scriptsize $z'\1$}}}   (\dzp-0.1,0.3) -- (\dzp,0);
\draw[rounded corners=5pt] (\dzp-\dr,0) -- (\dzp-\dr+0.1,-0.3)  --
node[fill=white] {\textcolor{black}{{\scriptsize $r\,\1$}}}  (\dzp-0.1,-0.3) -- (\dzp,0);
\draw[rounded corners=5pt] (\d1,0) -- (\d1+0.1,-0.3)  --
node[fill=white] {\textcolor{black}{{\scriptsize $r\,\1$}}}  (\d1+\dr-0.1,-0.3) -- (\d1+\dr,0);
\draw[rounded corners=5pt](\dz+\dzc-\dr,0) -- (\dz+\dzc-\dr+0.1,-0.3) -- 
node[fill=white] {\textcolor{black}{{\scriptsize $r\,\0$}}}  (\dz+\dzc-0.1,-0.3) -- (\dz+\dzc,0);
\end{tikzpicture}
\]
We have shown that $z^kz'\cdot\1$ is unbordered, that is, it is its own periodic root. The claim now follows from Lemma \ref{lem:4}.
\end{proof}

\begin{lemma}\label{lem:onlyone}
Let $u$ and $v$ be nonempty words such that both $uv$ and $vu$ are Lyndon. Then zero is the only cyclic occurrence of $u$ in $uv$.  
\end{lemma}
\begin{proof}
Without loss of generality, let $uv$ be $\less$-Lyndon word.
Suppose that $uv'$ is a conjugate of $uv$. From $uv \less uv'$ we deduce $v \less v'$, which is equivalent to $v'\LESS v$ since $|v|=|v'|$. From $vu$ being $\LESS$-Lyndon, we have $v\LESS v'$, and therefore $v=v'$. The claim now follows from the fact that a primitive word is not a nontrivial conjugate of itself. 
\end{proof}

\section{Fully bordered words}
We say that a binary word $w$ is \emph{fully bordered} if $|w|>1$, and there are exactly two unbordered conjugates of $w$. If $w$ is fully bordered and $uv, vu$ are the unbordered conjugates of $w$, then we will slightly abuse terminology and say that $(u,v)$ is a \emph{fully bordered pair}. Note the obvious but important fact that $(u,v)$ is a fully bordered pair if and only if the local period of $uv$ is less than $|uv|$ at all points except $0$ and $|u|$.

We adopt the following notation. Let $z$ be the periodic root of a nonempty word $w$. Then $w=z^kz'$, where $z'$ is a nonempty prefix of $z$, and we define $s_w:=z'$, $t_w:=(z')^{-1}z$ and $k_w:=k$, which yields $w=(s_wt_w)^{k_w}s_w$. Note that we assume that $s_w$ is not empty. This means that if $w=z^k$, then $s_w=z$, $t_w$ is empty, and $k_w=k-1$. In particular, $k_w=0$ if and only if $w$ is unbordered.

\begin{lemma}\label{lem1}
Let $(u,v)$ be a fully bordered pair. Then:
\renewcommand{\theenumi}{\alph{enumi}}
\begin{enumerate}
	\item $(v,u)$ is a fully bordered pair. \label{ch0}
	\item  $uv$, $vu$, $u^Rv^R$ and $v^Ru^R$ are Lyndon words. \label{ch1}
	\item $s_ut_u$ and $(t_us_u)^R$ are Lyndon words. \label{chlyn}
	\item If $k_u>1$, let $u'=(s_ut_u)^{k_u-1}s_u$. Then $s_{u'}=s_u$ and $t_{u'}=t_u$. \label{ch3}
\end{enumerate}
\end{lemma}
\begin{proof}
 \eqref{ch0} Follows directly from the definition.

\eqref{ch1} The definition of a fully bordered word implies that $uv$ is primitive and contains both letters $\0$ and $\1$. Therefore, there are at least two unbordered conjugates of $uv$, namely the Lyndon conjugates with respect to $\less$ and $\LESS$. This implies that $uv$ and $vu$ are Lyndon words. A word is bordered if and only if its reversal is. Hence $u^Rv^R$ and $v^Ru^R$ are the only unbordered conjugates of $(uv)^R$, which means that they are its two Lyndon conjugates. 

\eqref{chlyn} Follows from \eqref{ch1} by Lemma \ref{lem:4}.

\eqref{ch3} The word $u'$ has a period $|s_ut_u|$. It is the least period since $s_ut_u$ is unbordered.
\end{proof}

The key for a characterization of fully bordered words is the following theorem.

\begin{theorem}\label{thm:1}
	Let $(u,v)$ be a fully bordered pair such that $|v|\leq |u|$ and $|uv|>2$. Let $u'=(s_ut_u)^{-1}u$. Then either 
	\begin{itemize}
		\item $\left(u',v\right)$ is a fully bordered pair, or
		\item $v=v'u'v'$ and $\left(u,v'\right)$ is a fully bordered pair.
	\end{itemize}
\end{theorem}
\begin{proof}

Suppose, without loss of generality, that $\0$ is the first letter of $u$. Then $uv$ and $u^Rv^R$ are $\less$-Lyndon words, and $vu$ and $v^Ru^R$ are $\LESS$-Lyndon words. In particular, $\0$ is a suffix of $s_u$, and since $|u|\geq 2$, it follows that $k_u\geq 1$. 
 
Let $p$ be the longest common prefix of $v$ and $t_us_u$. Then $p$ is a proper prefix of $v$ since $|t_us_u|$ is not a period of $uv$. It is also a proper prefix of $t_us_u$ since $vu$ is unbordered. This implies, since $vu$ is $\LESS$-Lyndon, that $p\cdot \0$ is a prefix of $t_us_u$ and $p\cdot\1$ is a prefix of $v$. Symmetrically, we define $q$ as the longest common suffix of $v$ and $s_ut_u$ and observe that $\1\cdot q$ is a suffix of $v$ while $\0\cdot q$ is a suffix of $s_ut_u$ since $v^Ru^R$ is $\LESS$-Lyndon. \smallskip

\analyze{Claim 1.} We first show that $t_us_u$ is a $\LESS$-Lyndon word. This is trivial if $u$ is a power of $\0$.
Let $u$ contain $\1$ and let $z$ be the $\LESS$-Lyndon conjugate of $s_ut_u$. 
 Consider the last occurrence of $z$ in $up$. More precisely, let $up=u_1zu_2$ with $|u_2|<|z|$. Let $r$ be the local periodic root of $uv$ at the point $\abs{u_1z}$. Since $z$ is unbordered, we deduce that $|r|>|u_2|$ and therefore $u_2\cdot \1$ is a prefix of $r$. If $|r|\leq|u_1z|$, then $u_2\cdot\1$ is a factor of $u_1z$. This is a contradiction with $z$ being $\LESS$-Lyndon, since $u_2\cdot \0$ is a prefix of $z$. Therefore, $|r|>|u_1z|$ and $up\cdot \1$ is a factor of $rr$. Since $up\cdot\1$ is unbordered by Lemma \ref{lem:5}, we obtain that $|r|\geq |up\cdot\1|$ and $|rr|>|uv|$. Therefore, $uv$ is a factor of $rr$. Since $uv$ is unbordered, we obtain that $|uv|\leq |r|$. The definition of the local period now yields $|r|=|uv|$, and since $(u,v)$ is fully bordered, this implies $u_1z=u$, and $z=t_us_u$.
\smallskip

\analyze{Claim 2.}	We now show that $t_u$ is a prefix of $v$. Suppose, for the sake of contradiction, that $p<_p t_u$. Let $r$ be the local periodic root of $uv$ at the point $(s_ut_u)^{k_u}$.
Since $s_ut_u$ is unbordered, we have $s_up\cdot\1\leq_p r$.
If $|r|<\abs{(s_ut_u)^{k_u}}$, then $p\cdot \1$ is a factor of $(s_ut_u)^{k_u}$, a contradiction with $p\cdot\0$ being a prefix of the $\LESS$-Lyndon word $t_us_u$. If $|r|>\abs{(s_ut_u)^{k_u}}$, then $rr$ contains $up\cdot \1$. Since $up\cdot \1$ is unbordered, we obtain that $|r|\geq |up\cdot \1|$ and $|rr|>|uv|$. As above, this implies $|r|=|uv|$ and the local period of $uv$ at the point $(s_ut_u)^{k_u}$ is $|uv|$, a contradiction.
\smallskip

\analyze{Claim 3.} By symmetry we deduce from previous claims that $t_u$ is a suffix of $v$ and $(s_ut_u)^R$ is the Lyndon conjugate of $(t_us_u)^R$.

\smallskip
\analyze{Claim 4.} We are going to show that the local periodic root of $s_ut_u$ at a point $0 < m < |s_u|$ is short. To be precise, we claim that if $s_u=s_1s_2$, with nonempty $s_1$ and $s_2$, then there is a word $r$ which is both a prefix of $s_2 p$ and a suffix of  $qs_1$.
To prove this, let $r$ be the local periodic root of $s_ut_u$ at the point $|s_1|$.
Suppose that the claim does not hold. Then the local periodic root $r'$ of $uv$ at the point $\abs{us_2^{-1}}$ contains either $p\cdot \1$ or $\1\cdot q$ as a factor. Indeed, if $r'$ is both a prefix of $s_2p$ and a suffix of $qs_1$, then $r'=r$ and the claim holds. Note that $p\cdot \1$ is not a  factor of $qup$, since $p\cdot \0$ is a prefix of the $\LESS$-Lyndon word $t_us_u$. 
Similarly, $\1\cdot q$ is not a  factor of $qup$, since $(\0\cdot q)^R$ is a prefix of the $\LESS$-Lyndon word $(s_ut_u)^R$.
Therefore $\1\cdot qup\cdot \1$ is a factor of $r'r'$, which once more implies $|r'r'|>|uv|$ and $\abs{r'}=\abs{uv}$, a contradiction. The proof that $r$ is a suffix of $qs_1$ is similar, using the point $\abs{s_1}$ of $uv$.

\smallskip

	Consider now two cases corresponding to the two possibilities in the theorem. 
		
		\analyze{Case 1.} Suppose that $(s_ut_u)^{k_u}$ is not a factor of $v$. 
We show that the pair $\left(u',v\right)$	is fully bordered by comparing the local periodic roots of $uv$, $s_ut_u$ and $u'v$ at corresponding points.

The word $uv$ contains exactly two occurrences of  $(s_ut_u)^{k_u}$, namely $0$ and $|s_ut_u|$. This follows from the assumption of Case 1, from the fact that $t_u$ is a prefix of $v$ while $t_us_u$ is not, and from the fact that $s_ut_u$ is unbordered, that is, not overlapping with itself. 
Therefore $(s_ut_u)^{k_u}$ is not a factor of any nontrivial local periodic root of $uv$, since otherwise it would have two non overlapping occurrences.

Consequently, the local periodic root of $uv$ at any point $|ut_u|<m<|uv|$ is also a nontrivial local periodic root of $u'v$ at the point $m-|s_ut_u|$.

For $|u|<m\leq |ut_u|$, observe that the local periodic root $r$ of $uv$ at $m$ is either a factor of $qut_u$, and then $|r|$  is at most the period of $u$, that is $\abs{s_ut_u}$, or $|r|>|qu|$, which implies $\abs{rr}>\abs{uv}$ and $\abs{r}=\abs{uv}$, a contradiction. 

Claim 2, Claim 3 and Claim 4 implies that the local periodic root of $u'v$ at a point $0<m'<|u|-|s_ut_u|$ is the same as the local periodic root of $uv$ at the point $m'+|s_ut_u|$, namely the same as the local periodic root of $s_ut_u$ at the point $m'\mod |s_ut_u|$. 

To complete the proof for this case, it remains to note that $u'v$ and $vu'$ are unbordered: if the shortest border of $u'v$ or $vu'$ is  shorter that $(s_ut_u)^{k_u}$, then it is also a border of $uv$ or $vu$; if it is longer, then $v$ contains $(s_ut_u)^{k_u}$.

	\smallskip
\analyze{Case 2.} Suppose, now, that $v$ contains a factor $(s_ut_u)^{k_u}$. Then it can be written as $v'u'v''$, where $v''=t_uv_2$ for some $v_2$. 
	By $|u|\geq|v|$, we have that $|v'v_2|\leq |s_u|$.  
			Since $s_ut_u$ is unbordered and $t_u$ is a prefix of $v$, it follows that $uv$ contains exactly $2k_u+1$ occurrences of $s_ut_u$, those visible in the factorization 
			$$uv=(s_ut_u)^{k_u+1}v_1(s_ut_u)^{k_u}v_2,$$ 
			where $v'=t_uv_1$, and $v_1$ is a nonempty word since $t_us_u$ is not a prefix of $v$.   
	This implies that $uv'=(s_ut_u)^{k_u+1}v_1$ is the periodic root of $uv'(s_ut_u)^{k_u}$. By Lemma \ref{lem:4}, it is a Lyndon word. 
	Therefore, the local periodic root $r$ of $uv$ at the point $\abs{(s_ut_u)^{k_u+1}v_1}$ has a proper prefix $(s_ut_u)^{k_u}$, otherwise $r$ is a border of the Lyndon word $(s_ut_u)^{k_u+1}v_1$. Since we know all occurrences of $s_ut_u$ in $uv$, we deduce that 
	$r=(s_ut_u)^{k_u}v_1=u'v'$.
	
If $|s_ut_urr|>|uv|$, then the second occurrence of $r$ overlaps with $s_ut_u$. Then $r$ is bordered, a contradiction.
Let $uv=s_ut_urrz$, where $|z|<|s_u|$ for length reasons. Suppose that $z$ is not empty. Since $s_ut_u$ is unbordered, the word $s_ut_u$ is not a prefix of $zs_ut_u$. Therefore we deduce $zs_ut_u\LESS s_ut_u$ from $s_ut_u\less zs_ut_u$, which follows from $uv$ being $\less$-Lyndon. We obtain a contradiction with $vu$ being $\LESS$-Lyndon, since $v's_ut_u$ is a prefix of $v$ and $v'zs_ut_u$ is a factor of $vu$.

Therefore, $z$ is empty, and $uv=s_ut_urr=uv'u'v'$ as claimed.

\smallskip
We show that $\left(u,v'\right)$ is fully bordered. 
We have seen that $uv'$ is a  Lyndon word, in particular unbordered, which implies that $p$ is a proper prefix of $v'$. Symmetrically, we can show that $q$ is a suffix of $v'$. 
Claim 4 now implies that the local periodic root of $uv'$ at a point $0<m<|u|$ is the same as the local periodic root of $uv$ at the point $m$.

Let $|u|< m \leq |up|$. The argument is similar as in Case 1: The local periodic root $r$ of $uv$ at the point $m$ has length at most $|u|$, since $|u|\geq |v|$. Therefore, $r$ has a period $|s_ut_u|$, which implies $|r|\leq |s_ut_u|$ since $r$ is unbordered. Therefore, $r$ is also the local periodic root of $uv'$ at the point $m$. Similarly, we can show that the local periodic root of $uv'$ at a point $|uv'|-m$, where $0 < m \leq |q|$, is the same as the local periodic root of $uv$ at the point $|uv| - m$.

Let $r$ be the local periodic root of $uv$ at a point  $|up|< m < |uv'|- |q|$. If $|r|\leq |uv'u'p|-m$, then $r$ is also a local periodic root of $uv'$ at the point $m$ (see Figure \ref{fig}).
\begin{figure}%
\begin{tikzpicture}
[tecka/.style={circle,inner sep=1.3pt, fill}, scale=0.5
]
\def\da{0} \def\db{2} \def\du{7} \def\dut{8} \def\dm{9} \def\dtuv{11} \def\duv{12} \def\duvu{19} \def\duvub{17}
\def\lev{-3} \def\drr{15} \def\dr{3}
\def\dss{5} \def\dtt{0.5} \def\dvc{2.5} \def\dpp{1}
\def\du{2*\dss+\dtt} \def\duvc{\du +\dvc} \def\duv{\du+2*\dvc+\dss}  \def\duvcuc{\du+\dvc+\dss} \def\duvcu{2*\du+\dvc}
\def\dm{\du+\dpp+0.25} \def\delr{\dvc+\dss -0.8} \def\drrr{\dm  - 6.7} \def\drr{\dm+\delr}
\node[tecka] (a) at (0,0){};
\node[tecka] (b) at (0,\lev){};
\node[tecka] (u) at (\du,0){};
\node[tecka] (ub) at (\du,\lev){};
\node (uvc) at (\duvc,0){$\vert$};
\node (uvcuc) at (\duvcuc,0){$\vert$};
\node[tecka] (uvcb) at (\duvc,\lev){};
\node[tecka] (uv) at (\duv,0){};
\node[tecka] (uvu) at (\duvcu,\lev){};
\draw[thick] (a.center) -- (uv.center);
\draw[thick] (b.center) -- (uvu.center);
\draw[gray,rounded corners=5pt] (\du,0) -- (\du+0.1,0.5) -- 
node[fill=white, inner sep=1pt] {\textcolor{black}{{\scriptsize $p$}}} (\du+\dpp-0.1,0.5) --(\du+\dpp,0);
\draw[gray,rounded corners=5pt] (\duvc-\dpp,0) -- (\duvc-\dpp+0.1,0.5) -- 
node[fill=white, inner sep=1pt] {\textcolor{black}{{\scriptsize $q$}}} (\duvc-0.1,0.5) --(\duvc,0);
\draw[gray,rounded corners=5pt] (\duvcuc,0) -- (\duvcuc+0.1,0.5) -- 
node[fill=white, inner sep=1pt] {\textcolor{black}{{\scriptsize $p$}}} (\duvcuc+\dpp-0.1,0.5) --(\duvcuc+\dpp,0);
\draw[gray,rounded corners=5pt] (\duv-\dpp,0) -- (\duv-\dpp+0.1,0.5) -- 
node[fill=white, inner sep=1pt] {\textcolor{black}{{\scriptsize $q$}}} (\duv-0.1,0.5) --(\duv,0);
\node[circle, fill=gray!30,inner sep=0] at (\dm,1.353+\lev) {{\small $m$}};
\node at (\dm,0.7+\lev) {$\downarrow$};
\node at (\dm,2+\lev) {$\uparrow$};
\draw[gray,rounded corners=5pt] (\drrr,0) -- (\drrr+0.1,-0.6) -- 
node[fill=white, inner sep=1pt] {\textcolor{black}{{\scriptsize $r$}}} (\dm-0.1,-0.6) --(\dm,0);
\draw[gray,rounded corners=5pt] (\drr,0) -- (\drr-0.1,-0.6) -- 
node[fill=white, inner sep=1pt] {\textcolor{black}{{\scriptsize $r$}}} (\dm+0.1,-0.6) --(\dm,0);
\draw[gray,rounded corners=5pt] (0,0) -- (0.1,0.6) -- 
node[fill=white, inner sep=1pt] {\textcolor{black}{{\scriptsize $u$}}} (\du-0.1,0.6) --(\du,0);
\draw[gray,rounded corners=5pt] (\duvc,0) -- (\duvc+0.1,0.6) -- 
node[fill=white, inner sep=1pt] {\textcolor{black}{{\scriptsize $u'$}}} (\duvcuc-0.1,0.6) --(\duvcuc,0);
\draw[gray,rounded corners=5pt] (0,\lev) -- (0.1,\lev-0.6) -- 
node[fill=white, inner sep=1pt] {\textcolor{black}{{\scriptsize $u$}}} (\du-0.1,\lev-0.6) --(\du,\lev);
\draw[gray,rounded corners=5pt] (\duvc,\lev) -- (\duvc+0.1,\lev-0.6) -- 
node[fill=white, inner sep=1pt] {\textcolor{black}{{\scriptsize $u$}}} (\duvcu-0.1,\lev-0.6) --(\duvcu,\lev);
\draw[gray,rounded corners=5pt] (\du,\lev) -- (\du+0.1,\lev-0.6) -- 
node[fill=white, inner sep=1pt] {\textcolor{black}{{\scriptsize $v'$}}} (\duvc-0.1,\lev-0.6) --(\duvc,\lev);
\end{tikzpicture}
\caption{}
\label{fig}%
\end{figure}
Suppose that $|r|$ is greater than that. Then the word $y=\1\cdot q u' p\cdot \1$ is a factor of $r$. We check all the three cyclic occurrences   of $u'$ in $uv$ (namely $0$, $|s_ut_u|$ and $|uv'|$), and verify that there is only one cyclic occurrence of $y$ in $uv$; a contradiction.
\end{proof}

Theorem \ref{thm:1} motivates the following inductive definition. 
Let $\+F$ be the smallest subset of $\{\0,\1\}^*\times \{\0,\1\}^*$ satisfying the following conditions:
\begin{enumerate} 
	\item $(\0,\1)\in\+F$
  \item If $(u,v)\in \+F$, then $(v,u)\in \+F$. \label{def2}
	\item Let $(u,v)\in \+F$. Then $(s_ut_uu,v)\in \+F$. \label{def3}
	\item Let $(u,v)\in\+F$, and let $y$ be a border of $u$ such that $|t_v|<|y|$. Then $(u,vyv)\in \+F$. \label{def4} 
\end{enumerate}

\bigskip
Pairs in $\+F$ have the following properties.
\begin{lemma}\label{lem:if}
Let $(u,v)\in \+F$.
\renewcommand{\theenumi}{\Alph{enumi}}
\begin{enumerate}
	\item Both $uv$ and $vu$ are unbordered. \label{cA}
	\item  $t_u$ is both a prefix and a suffix of $v$. \label{cB}
	\item If $y$ is a border of $u$ such that $|t_v|<|y|$, then $(y,v)\in\+F$. \label{cD}
	\item If $t_u$ is empty, then $s_u\in\{\0,\1\}$; otherwise $(s_u,t_u) \in \+F$. \label{cC}
\end{enumerate}
\end{lemma}

\begin{proof}
We proceed by induction on the structure of $\+F$.  
All claims are true for $(\0,\1)$ and $(\1,\0)$. 

Suppose now that $(u,v), (v,u)\in\+F$. 
We shall prove the theorem for the four pairs $(s_ut_uu,v)$, $(v,s_ut_uu)$, $(u,vyv)$ and $(vyv,u)$ where $y$ is a border of $u$ such that $|t_v|<|y|$. Observe that $s_{s_ut_uu}=s_u$ and $t_{s_ut_uu}=t_u$. By \eqref{cD} for $(u,v)$ and \eqref{cA} for $(y,v)$ we have that $vy$ is unbordered which implies that $|vy|$ is the least period of $vyv$. Therefore $s_{vyv}=v$ and $t_{vyv}=y$.

\eqref{cA} For all four pairs, the claim follows easily from $uv$ and $vu$ being unbordered. Consider for example $vyvu$ and suppose it has the shortest border $r$. If $|r|\leq |vy|$, then $r$ is also a border of $vu$. If $|vy|<r\leq |vyv|$, then $r$ is not unbordered. If $r=vyvu'$ with $u'\leq_p u$, then $vu'$ is a border of $vu$.

\eqref{cB} For all four pairs, the claim follows directly from \eqref{cB} for $(u,v)$ and $(v,u)$.

\eqref{cD} For $(s_ut_uu,v)$, note that the only border of $s_ut_uu$ which is not also a border of $u$ is $u$ itself (a border longer than $u$ would imply a period smaller than $|s_ut_u|$). The claim now follows from \eqref{cD} for $(u,v)$.

For $(v,s_ut_uu)$, let $y'$ be a border of $v$ such that $|t_u|<|y'|$. From \eqref{cD} for $(v,u)$, we obtain $(y',u)\in\+F$. Therefore also $(y',s_ut_uu)\in\+F$ by the definition of $\+F$.

For $(u,vyv)$, let $y'$ be a border of $u$ such that $|y|<|y'|$. Then $(y',v)\in \+F$ follows from \eqref{cD} for $(u,v)$. Since $y$ is a border of $y'$, and $|t_v|<|y|$, we have $(y',vyv)\in\+F$ from the definition of $\+F$.

For $(vyv,u)$, let $y'$ be a border of $vyv$ such that $\abs{t_u}<\abs{y'}$. From \eqref{cD} for $(u,v)$, we deduce that $vy$ is unbordered. This implies that $|y'|\leq |v|$. If $y'=v$, then the claim hold since $(v,u)\in \+F$. If  $|y'|<|v|$, then $y'$ is a border of $v$ such that $|t_u|<|y'|$ and the claim follows from \eqref{cD} for $(v,u)$.

\eqref{cC} Straightforward for pairs $(s_ut_uu,v)$, $(v,s_ut_uu)$ and $(u,vyv)$. For $(vyv,u)$, we have that $t_{vyv}=y$ is not empty, and we want to show that $(v,y)\in \+F$, since $(v,y)=(s_{vyv},t_{vyv})$. This follows from \eqref{cD} for $(u,v)$ and from item \eqref{def2} of the definition of $\+F$.
\end{proof} 

We are now ready for the main result.

\begin{theorem}\label{thm:main}
    The pair $(u,v)\in \{\0,\1\}^*\times \{\0,\1\}^*$ is fully bordered if and only if $(u,v)\in\+F$.
\end{theorem}

\begin{proof}
The ``only if'' part of Theorem \ref{thm:main} follows from Theorem \ref{thm:1} by induction on the length of the pair (that is, on $|uv|$) as follows. Both fully bordered pairs $(u,v)$ with $\abs{uv}\leq 2$ are in $\+F$. Consider a fully bordered pair $(u,v)$ of length $|uv|>2$ and let all fully bordered pairs of length less than $|uv|$ be in $\+F$. By symmetry expressed in Lemma \ref{lem1}\eqref{ch0}, we can suppose that $|v|\leq |u|$. 

Suppose first that $\left(u',v\right)$ is fully bordered, where $u'=(s_ut_u)^{-1}u$. Then $(u',v), (v,u')\in \+F$. If $k_u>1$, then $s_{u'}=s_u$ and $t_{u'}=t_u$ by Lemma \ref{lem1}\eqref{ch3}, which implies $(u,v)\in \+F$ by item \eqref{def3} of the definition of $\+F$.
If $k_u=1$, then $u'=s_u$. Lemma \ref{lem:if}\eqref{cB} implies that $t_{u'}$ is a prefix of $v$. Also $t_u$ is a prefix of $v$ by Claim 2 in the proof of Theorem \ref{thm:1}. Therefore, $t_u$ and $t_{u'}$ are prefix comparable. The word $s_ut_u$ is Lyndon by Lemma \ref{lem1}\eqref{chlyn}, and $t_us_u$ is Lyndon by Claim 1 in the proof of Theorem \ref{thm:1}. Since $t_{u'}$ is a factor of $s_u$, Lemma \ref{lem:onlyone} now implies that $t_u$ is not a prefix of $t_{u'}$. Therefore, $t_{u'}$ is a proper prefix of $t_u$. Item \eqref{def4} of the definition of $\+F$ applied to the pair $(v,u')$ with $y=t_u$ yields $(v,u't_uu')=(v,u)\in\+F$. Therefore also $(u,v)\in\+F$.

Suppose now that $v=v'(s_ut_u)^{-1}uv'$ and $(u,v')$ is fully bordered. Then $(u,v')\in\+F$.
From Lemma \ref{lem:if}\eqref{cB} we have $|t_u|<|v'|$, which together with $|v|\leq|u|$ implies $|v'|<|s_u|$. Therefore also $|t_{v'}|<|y|$, where $y=(s_ut_u)^{-1}u$. Item \eqref{def4} of the definition of $\+F$ again implies $(u,v'yv')=(u,v)\in\+F$.

\medskip
The pairs $(\0,\1)$ and $(\1,\0)$ are fully bordered. To prove the ``if'' part by induction, let $(u,v)\in \+F$, and let all pairs in $\+F$ shorter than $(u,v)$ be fully bordered.

 Consider the local periodic root $r$ of $s_ut_u$ at a point $0< m<|s_u|$. If $t_u$ is empty, then $s_u\in\{\0,\1\}$ by Lemma \ref{lem:if}\eqref{cC} and $r=s_u$. If $t_u$ is nonempty, then Lemma \ref{lem:if}\eqref{cC} and the induction assumption imply that $(s_u,t_u)$ is fully bordered. Therefore $|r|<|s_ut_u|$, and Lemma \ref{lem:onlyone} implies that $t_u$ is not a factor of $r$. This yields a factorization $t_us_ut_u=u_1rru_2$ where $|u_1r|=|t_u|+m$. Since $t_u$ is both a prefix and a suffix of $v$ by Lemma \ref{lem:if}\eqref{cB}, we deduce that $r$ is the local periodic root of $uv$ at any point $0<m'<|u|$ such that $m=m'\mod |s_ut_u|$. We also easily observe that the local periodic root of $uv$ at a point $|s_u|\leq m\leq |u| - |s_u|$ is the same as the local periodic root of $|s_ut_u|$ at the point $m \mod |s_ut_u|$. We have shown that local periodic roots of $uv$ at all points $0<m<|u|$ are short.

Repeating the same argument for the word $v$ completes the proof that $(u,v)$ is a fully bordered pair. 
\end{proof}

The definition of $\+F$ and Lemma \ref{lem:if} yield a good insight into the structure of fully bordered words. One of straightforward corollaries is the answer to the question that served as a motivation for our research.

\begin{theorem}
Let $(u,v)$ be a fully bordered pair. Then both $u$ and $v$ are palindromes.
\end{theorem} 
\begin{proof}
Follows directly by induction from the definition of $\+F$. It is enough to observe that $s_u$, $t_u$, and any border of $u$ are palindromes if $u$ is a palindrome.
\end{proof}

\section*{Acknowledgements}
We want to thank Luca Zamboni for introducing us to this problem.
Furthermore, the second author expresses his gratitude to Tero Harju and Luca Zamboni for numerous interesting discussions on the subject.
We also thank Jana Hadravov\'a for many suggestions improving the manuscript.


\begin{thebibliography}{1}

\bibitem{pallength}
A.~Frid, S.~Puzynina, and L.~Zamboni.
\newblock On palindromic factorization of words.
\newblock {\em Advances in Applied Mathematics}, 50(5):737 -- 748, 2013.

\bibitem{density}
T.~Harju and D.~Nowotka.
\newblock Density of critical factorizations.
\newblock {\em Theor. Inform. Appl.}, 36(3):315--327, 2002.

\bibitem{bordercorr}
T.~Harju and D.~Nowotka.
\newblock Border correlation of binary words.
\newblock {\em J. Comb. Theory, Ser. A}, 108(2):331 -- 341, 2004.

\bibitem{DuvalDirk}
T.~Harju and D.~Nowotka.
\newblock Periodicity and unbordered words: {A} proof of the extended {D}uval
  conjecture.
\newblock {\em J. {ACM}}, 54(4), 2007.

\bibitem{Duval}
{\v{S}}.~Holub.
\newblock A proof of the extended {D}uval's conjecture.
\newblock {\em Theoret. Comput. Sci.}, 339(1):61--67, 2005.

\bibitem{EhrSil}
{\v{S}}.~Holub and D.~Nowotka.
\newblock The {E}hrenfeucht-{S}ilberger problem.
\newblock {\em J. Comb. Theory, Ser. A}, 119(3):668 -- 682, 2012.

\bibitem{jeff}
J.~Shallit.
\newblock {\em A Second Course in Formal Languages and Automata Theory}.
\newblock Cambridge University Press, 2008.


\end{thebibliography}
\end{document}